%% file: moderate_deviations.tex
\documentclass[10pt, conference]{IEEEtran}

\input{preamble}

\title{Moderate-Deviations of Lossy Source Coding for Discrete    and Gaussian    Sources}

\IEEEoverridecommandlockouts

\author{\IEEEauthorblockN{Vincent Y. F. Tan }

\vspace{0.5em}

\IEEEauthorblockA{Institute for Infocomm Research, A*STAR   (Email:
\url{tanyfv@i2r.a-star.edu.sg})}
\IEEEauthorblockA{Department  of Electrical and Computer Engineering, National University of Singapore}
}

\begin{document}

\maketitle

\begin{abstract}
We study the moderate-deviations (MD) setting for lossy source coding of stationary  memoryless sources. More specifically, we derive fundamental compression limits of source codes whose rates are $R(D)\pm\epsilon_n$, where $R(D)$ is the rate-distortion function and $\epsilon_n$ is a   sequence that   dominates $\sqrt{1/n}$. This MD setting is complementary to the  large-deviations and central limit settings and was studied by Altug and Wagner for the channel coding setting. We show, for finite alphabet and Gaussian sources, that as in the central limit-type results, the so-called dispersion  for lossy source coding plays a fundamental role in the MD setting for the lossy source coding problem. 
\end{abstract}


\section{Introduction}\label{sec:intro}
Rate-distortion theory~\cite{Sha59} consists in finding the optimal compression rate for a source $X\sim P$ subject  to the condition that there exists a code which can reproduce the source to within a   distortion level $D$. The optimal compression rate for the distortion level $D$ is known as the {\em rate-distortion function} $R(P,D)$. This function can be expressed as the minimization of mutual information over test channels~\cite{Sha59}.

It is also of interest to study the  excess distortion probability for codes at rate $R>R(P,D)$. This is the probability that the average distortion between $X^n$ and its reconstruction $\hatX^n$ exceeds $D$. The  exact exponential rate of decay of this probability was derived by Marton~\cite{Marton74} for discrete memoryless sources (DMSs). This was     extended to Gaussian~\cite{Ihara00} and general sources~\cite{Iriyama05}. These results belong to the theory of {\em large-deviations} (LD) and are reviewed in Section~\ref{sec:system}.

With the   revival of interest in second-order coding rates and dispersion analysis~\cite{PPV10, Hayashi09, TK12}, various researchers have also studied the fundamental limit of lossy compression subject to the condition that the probability of excess distortion is no larger than  $\epsilon>0$. In particular,  it was shown in~\cite{ingber11} and independently in~\cite{Kos11,Kos11b}    that 
\begin{equation}
R(n,D,\epsilon)\approx R(P,D)+\sqrt{\frac{V(P,D)}{n} } Q^{-1}(\epsilon) , \label{eqn:finite}
\end{equation}
where $R(n,D,\epsilon)$ is the optimal rate of compression of a  memoryless source  at blocklength $n$ and $V(P,D)$ is  known as the {\em dispersion} of the source. Eq.~\eqref{eqn:finite} holds true  for both discrete and Gaussian sources and belongs to the realm of {\em central limit theorem} (CLT)-style results. 

In this paper, we operate in a moderate-deviations (MD) regime  \cite[Section 3.7]{Dembo} that ``interpolates  between'' the LD   and CLT regimes. In particular, we   study the performance of source codes of rates $R_n = R(P,D)\pm\epsilon_n$ where $\epsilon_n$ is a  sequence that is asymptotically larger than $\sqrt{1/n}$ (cf.\ \eqref{eqn:finite}). Our results apply to both   finite alphabet and Gaussian sources but do not reduce to  the LD  or CLT  settings. Moreover, neither the LD nor CLT results specialize to our setting.  We show that the dispersion  $V(P,D)$ also plays a fundamental role in this MD setting. Besides studying the excess distortion probability, we also study the complementary probability (also termed the probability of correct decoding) for codes whose rates are {\em below} the rate-distortion function. Similarly,  the fundamental nature of the dispersion is revealed.

This work is inspired by the work on  MD in the context of channel coding~\cite{altug10,PV10}. It was shown  in~\cite{altug10} that for {\em positive} discrete memoryless channels (i.e., $W(y|x)>0$ for all $x,y$), the dispersion also governs the ``MD exponent''  
\begin{equation}
\lim_{n\to\infty}\frac{1}{n \epsilon_n^2 } \log  e(f_n,\varphi_n, W)  = -\frac{1}{2V(W)}  . \label{eqn:channel_coding_md}
\end{equation} 
The direct part was proved by   considering the Taylor expansion of Gallager's random coding exponent. We also use this proof strategy. In~\cite{PV10}, several assumptions in~\cite{altug10} were relaxed and the relations between the MD and CLT were clarified. Concurrent to this work, Sason~\cite{Sason11} studied MD for binary hypothesis testing. Finally, we mention that He et al.~\cite{He09} studied the redundancy of the Slepian-Wolf problem which is also related to~\cite{ingber11, Kos11, Kos11b} and  to the current problem. 
\section{System Model and Basic Definitions}\label{sec:system}
Let $\calP(\calX)$ be the set of  probability mass functions  supported on the finite alphabet $\calX$. Let $\calP_n(\calX) \subset\calP(\calX)$ be the set of $n$-types. For a type $Q\in\calP_n(\calX)$, let $\calT_Q^n$ be the set of sequences $x^n$ of type $Q$, i.e., the type class.  The reproduction alphabet is denoted as $\hatcalX$. In addition, let $d:\calX\times\hatcalX \to \bbR^+$ be a   distortion measure such that for every $x\in\calX$, there exists an $\hatx_0 \in\hatcalX$ for which $d(x,\hatx_0 )=0$. The average   distortion is $d(x^n,\hatx^n):=\frac{1}{n}\sum_{i=1}^n d(x_i,\hatx_i)$.  For a function $f:\calA\to\calB$, the notation $\|f\|:=|f(\calA)|$ denotes the cardinality of its range.

A DMS $X^n\sim \prod_{i=1}^n P(x_i)$  is described at rate $R$ by an encoder. The decoder receives the description index over a noiseless link and generates a reconstruction sequence $\hatX^n\in\hatcalX^n$. We now remind the reader of  the rate-distortion problem.
\begin{definition} 
A {\em rate-distortion code} consists of (i) an encoder $f_n:\calX^n\to \calM_n$ and (ii)   a decoder  $\varphi_n:\calM_n \to \hatcalX^n$. The {\em rate} of the code is $R_n:=\frac{1}{n}\log |\calM_n |$. 
\end{definition} 
The {\em rate-distortion function} $R(P,D)$ is defined as the infimum of all numbers  $R$ for which there exists codes $\{(f_n, \varphi_n)\}_{n\in\bbN}$  for which the    {\em probability of excess distortion} 
\begin{equation}
e(  f_n,\varphi_n, P, D):=\bbP(d(X^n,\varphi_n(f_n(X^n))) > D ) 
\end{equation}
is arbitrarily small for sufficiently large blocklengths $n$.  The rate-distortion function~\cite{Sha59}  can be expressed as
\begin{equation}
R(P,D)=\min_{W : \bbE  [ d(X,\hatX)]\le D} I(P,W),
\end{equation}
where $\bbE [d(X,\hatX)]:= \sum_{x, \hatx} P(x)  W(\hatx|x) d(x,\hatx)$. Another fundamental quantity introduced by Ingber and Kochman~\cite{ingber11} is the {\em dispersion for lossy source coding}
\begin{equation}
V(P,D) :=\var_X [R'(X;P,D)], \label{eqn:defV}
\end{equation}
where $R'(x;P,D)= \frac{\partial}{\partial P(x)} R(P,D)$ for $x\in\calX$ is the partial derivative of the rate-distortion function w.r.t.\ $P(x)$ (assuming it exists).  In \eqref{eqn:defV}, the variance is taken  w.r.t.\  the distribution $P$ and $R'(X;P,D)$ is a function of the random variable~$X$. In fact, the term {\em dispersion} is usually an operational one but since it was shown in~\cite{ingber11} that the operational defintion coincides with the one in \eqref{eqn:defV}, we will abuse terminology and use the generic term dispersion for both quantities. 

We analyze $e(  f_n,\varphi_n, P, D)$ in the so-called MD regime  where the rate of the code $R_n:=\frac{1}{n}\log \|f_n\|  =  R(P,D)+\epsilon_n$ for some sequence $\epsilon_n$. Clearly, if $\epsilon_n\to 0$, then $R_n\to R(P,D)$. When the rate of the code $R$ is a constant strictly above  $R(P, D)$, Marton \cite{Marton74} showed that 
\begin{equation}
\lim_{n\to\infty}\frac{1}{n}\log e(  f_n,\varphi_n, P, D)=-F(P, R, D), \label{eqn:Fdef}
\end{equation}
where {\em Marton's exponent} is defined as 
\begin{equation}
F(P,R,D):=\min_{Q\in\calP(\calX):R(Q,D)\ge R} D(Q \, ||\, P). \label{eqn:marton}
\end{equation}
The exponent is positive for   $R>R(P,D)$. One can also consider  the {\em probability of correct decoding} $1-e(  f_n,\varphi_n, P, D)$. In~\cite[pp.~156]{Csi97}, it was shown that:
\begin{equation}
\lim_{n\to\infty} \frac{1}{n} \log \, (1-e(  f_n,\varphi_n, P, D)) = -G(P,R,D),\label{eqn:iri}
\end{equation}
where the {\em exponent for correct decoding} is 
\begin{equation}
G(P,R,D) :=\min_{Q\in\calP(\calX):R(Q,D)\le R} D(Q \, ||\, P).\label{eqn:iri_exp}
\end{equation}
The exponent is positive   for $R<R(P,D)$.  These limits and exponents   are Sanov-like LD results~\cite{Dembo}. We present MD versions of Marton's and Iriyama's  results where the normalizations in~\eqref{eqn:Fdef} and~\eqref{eqn:iri} need not be $\frac{1}{n}$. 
\section{Discrete Memoryless Sources (DMS)}
Our main result for a DMS with bounded distortion measure (i.e.  $d:\calX\times\hatcalX \to [0,d_{\max}]$) is stated as follows:
\begin{theorem} \label{thm:mdp} 
Let $\epsilon_n$ be any positive sequence satisfying
\begin{equation}
\lim_{n\to\infty}\epsilon_n = 0,\qquad \lim_{n\to\infty} \frac{n \epsilon_n^2 }{\log n }= \infty.  \label{eqn:seq}
\end{equation}
That is, $\epsilon_n =\omega( (\frac{\log n}{n})^{1/2})\cap o(1)$. Assume that $R(Q,D)$ is twice differentiable w.r.t.\  $Q$ in a   neighborhood of $P$ and   $V(P,D)>0$. There exists a rate-distortion code $\{(f_n,\varphi_n)\}_{n\in\bbN}$ with rates $\frac{1}{n}\log \|f_n\| \le R(P,D)+\epsilon_n$  such that 
\begin{equation}
\limsup_{n\to\infty}\frac{1}{n\epsilon_n^2} \log  e(  f_n,\varphi_n, P, D)\le - \frac{1}{2 V(P,D)}. \label{eqn:direct}
\end{equation}
Furthermore, every rate-distortion  code   $\{(f_n,\varphi_n)\}_{n\in\bbN}$ with rates $\frac{1}{n}\log \|f_n\| \le R(P,D)+\epsilon_n$     must satisfy
\begin{equation}
\liminf_{n\to\infty}\frac{1}{n\epsilon_n^2} \log  e(  f_n,\varphi_n, P, D)\ge - \frac{1}{2 V(P,D)}. \label{eqn:converse}
\end{equation}
\end{theorem}
Though somewhat ungainly, the log factor in~\eqref{eqn:seq} appears to be essential because the proof hinges on the method of types. So our analysis does not completely close the gap between the CLT and LD regimes.  This log factor is unnecessary in the Gaussian case as will be seen in Theorems~\ref{thm:mdp_gauss} and~\ref{thm:mdp_less}. Theorem~\ref{thm:mdp} means that if the dispersion $V(P,D)$ is small, the ``MD  exponent'' $(2V(P,D))^{-1}$ is large, corresponding to a faster decay in the excess distortion probability. This has the same interpretation as in the CLT regime~\eqref{eqn:finite}. As an example, for the Bernoulli source with Hamming distortion,  the dispersion can be computed as
\begin{equation}
V( \mathrm{Bern}(\alpha), D) = \alpha(1-\alpha) \log^2 \left(\frac{1-\alpha}{\alpha}\right) .
\end{equation}
The parameter that maximizes (resp.\ minimizes)   $V(P,D)$ is $\alpha \approx 0.0832$ (resp.\ $\alpha =0,0.5$).  Thus, the ``MD exponent'' is maximized when the source is deterministic or has maximum entropy. The proof uses the following  lemma, whose proof is essentially identical to that of~\cite[Theorem 8]{Tan11_IT}, where the divergence and the constraint set in~\eqref{eqn:marton} are approximated by a quadratic and an affine subspace  respectively.
\begin{lemma}\label{lemma:dev}
If the limit exists, Marton's exponent   satisfies
\begin{equation}
\lim_{\delta\to 0}\, \frac{F(P, R(P,D)+\delta, D) }{\delta^2} = \frac{1}{2 V(P,D)}. \label{eqn:taylor_marton}
\end{equation}
\end{lemma}
In the sequel, we assume that the limit in~\eqref{eqn:taylor_marton} exists. Otherwise, the results are modified accordingly by considering the upper and lower limits in~\eqref{eqn:taylor_marton} and replacing the dispersion by its upper and lower limit versions.  We   first prove the direct part of Theorem~\ref{thm:mdp} in~\eqref{eqn:direct} followed by the converse in~\eqref{eqn:converse}.
\begin{proof}
The code construction proceeds along the lines of that in~\cite{ingber11}. Fix  a sequence $\epsilon_n$ satisfying~\eqref{eqn:seq}.  From the refined type covering lemma by Berger (stated in~\cite{Yu93}), for every type $Q\in\calP_n(\calX)$ there exists a set $\calC_{Q}$ that completely $D$-covers $\calT_Q^n$ (i.e., for every $x^n\in\calT_Q^n$ there exists an $\hatx^n\in \calC_Q$ such that $d(x^n,\hatx^n)\le D$) and $\calC_Q$ has rate 
\begin{equation}
\frac{1}{n}\log |\calC_Q|\le R(Q,D)+J (|\calX|,|\hatcalX|) \frac{\log n}{n}. \label{eqn:berger}
\end{equation}
where $J$ is some function of the size of the alphabets.  Consider the set $\calC$ that that is the union of all sets that $D$-cover the  types $Q\in \calU_n(D,\epsilon_n)$, defined as 
\begin{align}
\calU_n(D,\epsilon_n):= \big\{Q\in &  \calP_n(\calX):R(Q, D) <  R(P,D)+\epsilon_n' , \nn\\*
&\|Q-P\|_1\le \epsilon_n/\sqrt{V(P,D)}\big\}. 
\end{align}
where $\epsilon_n':= \epsilon_n-J (|\calX|,|\hatcalX|) \frac{\log n}{n} - |\calX|\frac{\log (n+1)}{n}$. The second constraint on the $\ell_1$ distance of the type  $Q$ to the true distribution $P$ is to ensure that $R(\fndot, D)$ is differentiable. This is also done in \cite[Theorem 4]{He09}. Note that if $\epsilon_n$ satisfies \eqref{eqn:seq} so does $\epsilon_n'$. Now, consider the size  of $\calC$:
\begin{align}
 |\calC|&=\sum_{Q\in\calP_n(\calX): R(Q, D) <  R(P,D)+\epsilon_n'} |\calC_Q| \nn \\
& \le  (n+1)^{|\calX|} \exp\left[  n\left(R(Q^*, D)+J (|\calX|,|\hatcalX|) \frac{\log n}{n} \right)\right] \nn \\
& \le  \exp\left[ n\left(R(P, D)+\epsilon_n  \right) \right]   \label{eqn:size_code}
\end{align}
The first inequality applies~\eqref{eqn:berger} and the type counting lemma. Furthermore, $Q^*$ is the dominating type. The second inequality applies the definitions of $\calU_n$  and  $\epsilon_n'$. Take $f_n$ to be the function that maps a sequence $x^n\in\calX^n$ with type $P_{x^n}$ to a predefined  index in $\calC = \cup_{Q\in\calU_n}\calC_Q$ and take $\varphi_n$ to be the function that maps the index to the reproduction sequence in $\calC_{P_{x^n}}$ that $D$-covers $x^n$.  Now, we evaluate the error probability, which is the $P^n$-probability of the types not in $\calU_n(D,\epsilon_n)$. Consider, 
\begin{align}
\bbP( & R(P_{X^n}, D)\ge R(P,D)+\epsilon_n' )\nn\\
&\le\sum_{Q\in\calP_n(\calX):R(Q,D)\ge R(P,D)+\epsilon_n'} P^n\left(\calT_Q^n\right) \nn\\
&\le\sum_{Q\in\calP_n(\calX):R(Q,D)\ge  R(P,D)+\epsilon_n'} \exp(-nD(Q\,||\, P)) \nn\\
&\le (n+1)^{|\calX|}  \exp[ -n F(P, R(P,D)+\epsilon_n', D)  ], \label{eqn:prev1}
\end{align}
where we applied the type counting lemma and the definition of Marton's exponent in the last line. Next, from \cite{Weiss03},  
\begin{eqnarray}
\bbP ( \|P_{X^n}-P\|_1> \epsilon_n/\sqrt{V } ) \le 2^{|\calX|} \exp \left[-n  \epsilon_n^2/(2V)  \right]. \label{eqn:prev2}
\end{eqnarray}
Combining  \eqref{eqn:prev1} and \eqref{eqn:prev2} with the union bound, 
\begin{align}
& e(  f_n,  \varphi_n, P, D) \nn\\
 &\le 2\exp\left[ -n \left( \frac{\epsilon_n'^2}{2V(P,D)} - o(\epsilon_n'^2)- \frac{|\calX|\log (n+1)}{n} \right)\right],\nn
\end{align}
where  we invoked Lemma~\ref{lemma:dev} with $\epsilon_n'=o(1)$ in the role of $\delta$. Now, we take the logarithm and normalize  by $n\epsilon_n^2$ to assert the achievability part of the theorem in \eqref{eqn:direct}. Note that we used the fact that $\frac{\log n}{n\epsilon_n^2}\to 0$. 

Now for the converse, we fix a code $\{(f_n,\varphi_n)\}_{n\in\bbN}$ of rate $R_n=\frac{1}{n}\log \|f_n\|\le R(P,D)+\epsilon_n$ and observe that 
\begin{eqnarray}
 e(  f_n,  \varphi_n, P, D) \ge \bbP(d(X^n,\hatX^n)>D|\calE_{\Psi_n})\bbP(\calE_{\Psi_n}). \label{eqn:joint_lb}
\end{eqnarray}
where the event $\calE_{\Psi_n}:=\{ R(P_{X^n}, D)\ge R_n+\Psi_n\}$ and $P_{X^n}$ is the type of   $X^n$. From the converse of the type covering lemma~\cite[Lemma~3]{Zhang97}, for any type $Q\in\calP_n(\calX)$ such that $R(Q,D)>R$, the  fraction of   $\calT_Q^n$ that is covered by any set is no greater than $\exp[-n(R(Q,D)-R+K(|\calX|,|\hatcalX|)\frac{\log n}{n} )]$. Hence,  the first term above can be bounded as 
\begin{align}
\bbP(&d(X^n,\hatX^n)>D|\calE_{\Psi_n})  \nn\\
& \ge 1-\exp\left[-n \left( \Psi_n+ K(|\calX|, |\hatcalX|) \frac{\log n}{n} \right) \right]   \label{eqn:lower_}
\end{align}
Put $\Psi_n := (K(|\calX|, |\hatcalX|)+1) \frac{\log n}{n}$. Then, \eqref{eqn:lower_} yields
\begin{equation}
\bbP(d(X^n,\hatX^n)>D|\calE_{\Psi_n})\ge 1-\frac{1}{n}\ge\frac{1}{2}. \label{eqn:half}
\end{equation}
Hence, it remains to bound the second term in~\eqref{eqn:joint_lb}. Let $\epsilon_n' := \epsilon_n+\Psi_n$ and consider,
\begin{align}
  P^n(\calE_{\Psi_n})   &=   \bbP(R(P_{X^n}, D)-R(P,D)\ge R-R(P,D)\!+\!\Psi_n)\nn \\
&\ge \bbP(R(P_{X^n}, D)-R(P,D)\ge \epsilon_n+\Psi_n) \nn\\
&=\sum_{Q\in\calP_n(\calX): R(Q,D)\ge R(P,D)+\epsilon_n'}    P^n(\calT_Q^n) \nn\\
&\ge \sum_{Q\in\calP_n(\calX): R(Q,D)\ge R(P,D)+\epsilon_n'}   \frac{\exp(-nD(Q\, ||\, P))}{(n+1)^{|\calX|}} \nn \\
&\ge (n+1)^{-|\calX|} \exp\left[-n D(Q^{(n)}\, ||\, P)  \right] \label{eqn:first_term}
\end{align}
where the first inequality is from the definition of $R_n \le R(P,D)+\epsilon_n$ and in the last inequality we defined the type $Q^{(n)}:=\argmin_{Q\in\calP_n(\calX): R(Q,D)\ge R(P,D)+\epsilon_n'} D(Q\,||\, P)$.  In the appendix, we prove the following   key continuity statement.
\begin{lemma}  \label{lem:cont}
If $\epsilon_n'$ satisfies \eqref{eqn:seq}, the types $Q^{(n)}$ satisfy
\begin{equation}
\lim_{n\to\infty}\frac{D(Q^{(n)}\, ||\, P) }{  F(P,R(P,D)+\epsilon_n', D )}= 1 . \label{eqn:eta}
\end{equation}
\end{lemma}
Let $\eta>0$. For $n$ large enough, the ratio in~\eqref{eqn:eta} is smaller than $1+\eta$.  Uniting~\eqref{eqn:joint_lb} -- \eqref{eqn:eta} yields
\begin{align}
& e(  f_n,  \varphi_n, P, D) \nn\\
& \ge \frac{1}{2}(n+1)^{-|\calX|} \exp[-n (1+\eta)  F(P,R(P,D)+\epsilon_n', D ) ] \nn \\
& \ge \frac{1}{2}(n+1)^{-|\calX|} \exp\left[-n (1+\eta)  \left( \frac{\epsilon_n'^2}{2V(P, D )}  + o( \epsilon_n'^2 ) \right) \right]   . \nn
\end{align}
The last inequality is an application of Lemma~\ref{lemma:dev} with $\epsilon_n'=o(1)$ in the role of $\delta$. Now, we take the logarithm and normalize  by $n\epsilon_n^2$ to establish the converse  noting that $\eta$ is arbitrary,   $\Psi_n= O(\frac{\log n}{n})$ and $\frac{\log n}{n\epsilon_n^2}\to 0$. The latter allows us to assert that $\epsilon_n'/\epsilon_n \to 1$.  \end{proof}
Note that the {\em multiplicative} nature of \eqref{eqn:eta} is necessary to establish Theorem~\ref{thm:mdp}.  The analysis for the probability of correct decoding  $1-e(  f_n,\varphi_n, P, D)$ in the MD regime is analogous and is stated in the following:
\begin{theorem} \label{thm:mdp_less_g} 
Let $\epsilon_n$ be any positive sequence satisfying~\eqref{eqn:seq}. Assume that $R(Q,D)$ is twice differentiable w.r.t.\ $Q$ in a   neighborhood of $P$ and   $V(P,D)>0$ There exists a rate-distortion code $\{(f_n,\varphi_n)\}_{n\in\bbN}$ with rates $\frac{1}{n}\log \|f_n\| \ge R(P,D)-\epsilon_n$   such that 
\begin{equation}
\liminf_{n\to\infty}\frac{1}{n\epsilon_n^2} \log \, (1-e(  f_n,\varphi_n, P, D) )\ge - \frac{1}{2 V(P,D)}. \label{eqn:direct_correct}
\end{equation}
Furthermore, every rate-distortion  code   $\{(f_n,\varphi_n)\}_{n\in\bbN}$ with rates $\frac{1}{n}\log \|f_n\| \ge R(P,D)-\epsilon_n$      must satisfy
\begin{equation}
\limsup_{n\to\infty}\frac{1}{n\epsilon_n^2} \log \, (1-e(  f_n,\varphi_n, P, D) )\le - \frac{1}{2 V(P,D)}. \label{eqn:converse_correct}
\end{equation}
\end{theorem}
\begin{proof} Similar to Theorem~\ref{thm:mdp}. \end{proof}
\section{Quadratic Gaussian Source Coding}
We now turn our attention to the quadratic Gaussian setting where $X^n$ is a length-$n$ vector whose entries are identically distributed as zero-mean Gaussians with variance $\sigma^2$. The distortion measure is $d(x,\hatx):=(x-\hatx)^2$. It is known \cite{Sha59} that in this case, the rate-distortion function takes the form
\begin{equation}
R(\sigma^2,D) = \frac{1}{2}\log \max\left\{1, \frac{\sigma^2}{D} \right\}. \label{eqn:rd_gauss}
\end{equation}
Furthermore, Ihara and Kubo \cite{Ihara00} showed that the analogue of Marton's exponent in \eqref{eqn:marton} also holds in the Gaussian setting. Indeed, it is shown that the excess distortion exponent is 
\begin{equation}
F(\sigma^2, R, D) = \frac{1}{2} \left[ \frac{D}{\sigma^2} e^{2R} -1- \log \left( \frac{D}{\sigma^2} e^{2R} \right)\right]. \label{eqn:ihara}
\end{equation}
whenever $R> R(\sigma^2, D)$ and zero otherwise. The exponent for correct decoding $G(\sigma^2,R,D)$ takes the same form as in~\eqref{eqn:ihara} when $R<R(\sigma^2,D)$ and zero otherwise.  In this case, it is easy to show by direct differentiation of $F(\sigma^2, R, D)$  (or $G(\sigma^2, R, D)$) that the  dispersion for lossy source coding is  
\begin{equation}
V(\sigma^2, D) = \frac{1}{2}, \label{eqn:disp_gauss}
\end{equation}
for all $\sigma^2$ and all $D$. In analogy to Theorem~\ref{thm:mdp}, we have the following in the quadratic Gaussian setting:
\begin{theorem} \label{thm:mdp_gauss} 
Let $\epsilon_n$ be any positive sequence satisfying 
\begin{equation}
\lim_{n\to\infty} \epsilon_n = 0,\qquad  \lim_{n\to\infty} {n \epsilon_n^2 }=\infty.\label{eqn:seq2}
\end{equation}
There exists a rate-distortion code $\{(f_n,\varphi_n)\}_{n\in\bbN}$ with rates $\frac{1}{n}\log \|f_n\| \le R(\sigma^2,D)+\epsilon_n$  such that 
\begin{equation}
\limsup_{n\to\infty}\frac{1}{n\epsilon_n^2} \log  e(  f_n,\varphi_n,\sigma^2, D)\le - 1  . \label{eqn:direct_g}
\end{equation}
Furthermore, every rate-distortion  code   $\{(f_n,\varphi_n)\}_{n\in\bbN}$ with rates $\frac{1}{n}\log \|f_n\| \le R(P,D)+\epsilon_n$    must satisfy
\begin{equation}
\liminf_{n\to\infty}\frac{1}{n\epsilon_n^2} \log  e(  f_n,\varphi_n, \sigma^2, D)\ge -1. \label{eqn:converse_g}
\end{equation}
\end{theorem}
In contrast to the DMS case, the dispersion for the quadratic Gaussian case~\eqref{eqn:disp_gauss} is constant. Hence, the exponents in \eqref{eqn:direct_g} and \eqref{eqn:converse_g} are also constant. Also note from \eqref{eqn:seq2} that the requirement on  $\epsilon_n$ is less stringent than in the DMS case~\eqref{eqn:seq}. In particular, the log factor is no longer required. This is because the method of types is not used in the proof.
\begin{proof}
Fix the sequence $\epsilon_n$. For the direct part, let us consider the set of ``empirical variances''
\begin{eqnarray}
\calU_n( D,\epsilon_n):=\left\{\hsigma^2: |R(\hsigma^2,D)-R(\sigma^2,D)| < \epsilon_n'\right\},
\end{eqnarray}
where  $\epsilon_n':= \epsilon_n - \frac{5\log n}{2n} - \frac{\log 6}{n}$. By using the definition of $R(\sigma^2,D)$ in \eqref{eqn:rd_gauss}, it is easy to see that $\hsigma^2\in \calU_n$ if and only if $e^{-2\epsilon_n'} < \hsigma^2/\sigma^2 < e^{2\epsilon_n'}$. We now use a  result by Verger-Gaugry~\cite[Theorem 1.2]{VG05}, which in our context, says that $6 n^{5/2} (\sigma^2 e^{2\epsilon_n'} /D)^{n/2}$ reconstruction points suffice to $D$-cover  length-$n$ vectors $x^n$ whose empirical variance $\frac{1}{n}\sum_i x_i^2\in \calU_n$. Hence, the size of the code  is bounded as 
\begin{align}
|\calC|\le 6 n^{5/2} (\sigma^2 e^{2\epsilon_n'} /D)^{n/2}\le \exp(n (R(\sigma^2, D) + \epsilon_n )) ,  
\end{align}
where  we used the definition of $\epsilon_n'$. Hence, the rate $R_n\le R(\sigma^2,D)+\epsilon_n$ as required. For the   probability of excess distortion, we have 
\begin{align}
& e(  f_n, \varphi_n, \sigma^2, D) = \bbP \left(  \frac{1}{n}\sum_{i=1}^n X_i^2 \notin \calU_n \right) \nn \\
 &\le \bbP\left(   \frac{1}{n}\sum_{i=1}^n X_i^2  > \sigma^2 e^{ 2\epsilon_n'})  \right) + \bbP\left(   \frac{1}{n}\sum_{i=1}^n X_i^2  < \sigma^2 e^{-2\epsilon_n'} \right) \nn \\
 &\le 4 \exp \left[ -\frac{n}{2} \left( e^{2\epsilon_n' }- 1 - 2\epsilon_n'  \right)\right] .   \label{eqn:e}
\end{align}
The first inequality is by the definition of $\calU_n$ and the union bound. The second is an application of the upper bound of Cram\'{e}r's theorem~\cite{Dembo} applied to the $\chi_1^2$-random variables $X_i^2/\sigma^2$. Now note from Taylor's theorem that  $e^{2\epsilon_n' }- 1 - 2\epsilon_n' = 2\epsilon_n'^2+o(\epsilon_n'^2)$. Taking the logarithm, normalizing by $n\epsilon_n^2$ and taking the upper limit of~\eqref{eqn:e} yields the desired result in~\eqref{eqn:direct_g}.

We now turn our attention to the converse. The gist of the proof follows from the converse in~\cite{Ihara00} but, as we shall see, the error probability analysis is more intricate. Fix codes of rates $R_n=\frac{1}{n}\log \|f_n\|\le R(\sigma^2,D)+\epsilon_n$.  Let the reproduction sequences be denoted as $\hatx^n(m), m  \in \calM_n$. Also, let $\calA_n := \cup_{m\in\calM_n} \calB_n(\hatx^n(m), \sqrt{D})$ where $\calB_n(c^n,r)$ is the $n$-dimensional ball centered at $c^n$ with radius $r$. Now, let $\gamma_n>0$ be such that $\Vol(\calB_n(0,\gamma_n))=\Vol(\calA_n)$. Clearly, $\Vol(\calA_n)\le |\calM_n| \Vol(\calB_n(0,\sqrt{D}))$. Since $R_n=\frac{1}{n}\log|\calM_n|$, 
\begin{equation}
e^{nR_n}\ge \frac{\Vol(\calA_n)}{ \Vol(\calB_n(0,\sqrt{D}))} = \frac{\Vol(\calB_n(0,\gamma_n) )}{ \Vol(\calB_n(0,\sqrt{D}))} = \left(\frac{\gamma_n}{\sqrt{D}}\right)^n .\nn
\end{equation}
Hence, we have $R(\sigma^2,D)+\epsilon_n\ge R_n\ge \frac{1}{2}\log \frac{\gamma_n^2}{D}$, i.e., 
\begin{equation}
\gamma_n\le \sigma^2 e^{2\epsilon_n}.\label{eqn:gam_bd}
\end{equation}
The probability of excess distortion can be lower bounded as:
\begin{align}
e(  f_n, \varphi_n, \sigma^2, D) = \bbP(X^n\notin\calA_n)  \ge \bbP( X^n\notin\calB_n(0,\gamma_n)). \nn
\end{align}
Now define  the random variables $Y_i:=X_i^2/\sigma^2$ and note that the $Y_i$'s are $\chi_1^2$-distributed. With this notation, and using~\eqref{eqn:gam_bd},  
\begin{align}
e(  f_n, \varphi_n, \sigma^2, D) \!\ge\! \bbP\left( \frac{1}{n}\sum_{i=1}^n Y_i \!>\! \frac{\gamma_n}{\sigma^2}  \! \right)\!\ge \!\bbP\left( \frac{1}{n}\sum_{i=1}^n Y_i\!> \! e^{2\epsilon_n}  \! \right) \! .\nn
\end{align}
Recall that for the $\chi_1^2$-distribution, the cumulant generating  function is $\Lambda(\theta)=-\frac{1}{2}\log (1-2\theta)$ and the rate function is $I(y)= \max_{\theta} \{\theta y - \Lambda(\theta)\} = \frac{1}{2}(y-1) -\frac{1}{2}\log y$. Furthermore, $\theta^*(y):=\frac{1}{2}(1-\frac{1}{y})$ is the maximizer. Using the standard  change of measure technique  for the lower bound in Cram\'{e}r's theorem (see proof of   \cite[Theorem 2.2.3]{Dembo}),
\begin{align}
e(  f_n, \varphi_n, \sigma^2, D) \ge \beta_n\exp  \left[-n I(e^{2\epsilon_n}) - \frac{n}{2} (1-e^{-2\epsilon_n}) \tau_n\right],  \nn
\end{align}
where $\beta_n:= \bbP ( \frac{1}{n}\sum_{i=1}^n \tilY_i  \in (e^{2\epsilon_n}, e^{2\epsilon_n}+\tau_n) )$ and $\tau_n$ is a sequence  to be chosen.  The random variables $\tilY_i$ have (tilted) distribution $q(\tily) := \exp[\theta^*(e^{2\epsilon_n})  \tily-\Lambda(\theta^*(e^{2\epsilon_n}))]p(\tily)$ where $p( \fndot)$ is the $\chi_1^2$ distribution of the $Y_i$'s. By the choice of $q(\fndot)$,  $\bbE_q [\tilY_i]=e^{2\epsilon_n}$. Put $\tau_n:=\zeta \epsilon_n$ for some $\zeta>0$. Then, 
\begin{align}
1-  \beta_n &\le \bbP\left(  \frac{1}{n}\sum_{i=1}^n \tilY_i\le e^{2\epsilon_n}\! \right)\!+\!\bbP\left(  \frac{1}{n}\sum_{i=1}^n \tilY_i\ge e^{2\epsilon_n} +\tau_n \! \right) \nn\\
&\le  \frac{1}{2}+ \frac{32}{ \sqrt{n}} + \frac{2}{n\tau_n^{2} }  = \frac{1}{2} + \frac{32}{ \sqrt{n}} +\frac{2}{ n\zeta^2\epsilon_n^2}, \nn
\end{align}
where in the second inequality, we applied the Berry-Ess\'{e}en theorem to the first term (the third moment of $\tilY_i$ is $15e^{-7\epsilon_n}$) and Chebyshev's inequality to the second. By~\eqref{eqn:seq2},  $\beta_n\to \frac{1}{2}$ from below. With this choice of $\tau_n$,   for $n$ sufficiently large, 
\begin{equation}
e(  f_n, \varphi_n, \sigma^2, D) \ge \frac{1}{4} \exp  \left[-n   \epsilon_n^2 (1 + \zeta+ o(1)  )\right] , \label{eqn:last_conv}
\end{equation}
where applied the facts $I(e^{2\epsilon_n})=\epsilon_n^2+ o(\epsilon_n^2)$ and $1-e^{-2\epsilon_n} = 2\epsilon_n +o(\epsilon_n)$. The converse in \eqref{eqn:converse_g} follows by taking the logarithm, normalizing by $n\epsilon_n^2$, taking  $n\to\infty$, and finally taking $\zeta\to 0$. 
\end{proof}
The MD setting for the probability of correct decoding of Gaussian sources can also analyzed analogously:
\begin{theorem} \label{thm:mdp_less} 
Let $\epsilon_n$ be any positive sequence satisfying~\eqref{eqn:seq2}.  There exists a rate-distortion code $\{(f_n,\varphi_n)\}_{n\in\bbN}$ with rates $\frac{1}{n}\log \|f_n\| \ge R(\sigma^2,D)-\epsilon_n$  such that 
\begin{equation}
\liminf_{n\to\infty}\frac{1}{n\epsilon_n^2} \log \, (1-e(  f_n,\varphi_n,\sigma^2, D) )\ge - 1. \label{eqn:direct_correct_g}
\end{equation}
Furthermore, every rate-distortion  code   $\{(f_n,\varphi_n)\}_{n\in\bbN}$ with rates $\frac{1}{n}\log \|f_n\| \ge R(\sigma^2,D)-\epsilon_n$    must satisfy
\begin{equation}
\limsup_{n\to\infty}\frac{1}{n\epsilon_n^2} \log \, (1-e(  f_n,\varphi_n, \sigma^2, D) )\le - 1. \label{eqn:converse_correct_g}
\end{equation}
\end{theorem}
\begin{proof} Similar to Theorem~\ref{thm:mdp_gauss} and uses ideas in~\cite{Ihara00}. \end{proof}
\section{Conclusion}
In this paper, we analyzed  the MD regime for lossy source coding. In analogy to~\eqref{eqn:channel_coding_md}, we showed for discrete sources that 
\begin{equation}
\lim_{n\to\infty}\frac{1}{n\epsilon_n^2}\log e(f_n,\varphi_n, P, D) =- \frac{1}{2V(P,D)} \label{eqn:conclu}
\end{equation}
and for Gaussian sources the RHS of~\eqref{eqn:conclu} is equal to $-1$ independent of the variance  $\sigma^2$  and the distortion level $D$. As in \cite{ingber11, Kos11, Kos11b}, this reveals that the fundamental nature of the dispersion    in the lossy source coding context. There are at least three avenues for  future research: (i) Can the results   be applied to, for instance, general sources as in \cite{Iriyama05}? (ii) Can similar analysis of the MD setting be applied to lossy source coding problems with {\em side information}, e.g., the Wyner-Ziv problem? (iii) What is the exact relationship between the MD and CLT regimes cf.\ \cite{PV10}?
\section*{Appendix: Proof of Lemma~\ref{lem:cont}}
\begin{proof}
The rate-distortion function is uniformly continuous. Specifically, $R(Q,D)-R(P,D)= O( \|Q-P\|_1 \log \|Q-P\|_1)$ \cite{Pal08}. Also,  $\min_{Q\in\calP_n(\calX) } \|Q-P\|_1\le |\calX|/n$ for any $P\in\calP(\calX)$~\cite[Lemma~2.1.2]{Dembo} so $\min_{Q\in\calP_n(\calX)}  R(Q,D)-R(P,D)  = O(\frac{\log n}{n})$ which is  asymptotically dominated  by  $\epsilon_n'  = \omega((\frac{\log n}{n})^{1/2})$. Thus, there exist  $n$-types in the regular-closed  set  $\{Q\in\calP (\calX):R(Q,D)-R(P,D)\ge \epsilon_n'\}$ for $n$ large.   Let  Marton's exponent be $D(Q_{\mathrm{M}}^{(n)}\,||\, P) =  F(P,R(P,D)+\epsilon_n', D )$.    Then, notice that
\begin{align}
\frac{D(Q^{(n)}\, ||\, P)}{D(Q_{\mathrm{M}}^{(n)}\, ||\, P)} = \frac{D(Q^{(n)}\, ||\, P)- D(Q_{\mathrm{M}}^{(n)}\, ||\, P)}{D(Q_{\mathrm{M}}^{(n)}\, ||\, P)}+1.  \label{eqn:fra}
\end{align}
The numerator of the first term  on the RHS in~\eqref{eqn:fra} is $O(\frac{1}{n})$ because $|D(Q^{(n)} \, ||\,  P) -D(Q_{\mathrm{M}}^{(n)} \,||\, P)| = O(\|Q^{(n)}-Q_{\mathrm{M}}^{(n)}\|_1)$ and $\|Q^{(n)}-Q_{\mathrm{M}}^{(n)}\|_1=O(\frac{1}{n})$. From Lemma~\ref{lemma:dev}, the denominator (Marton's exponent) scales as $\epsilon_n'^2/(2V(P,D))=\omega(\frac{\log n}{n})$. Thus, the first term in~\eqref{eqn:fra}  tends to zero and the ratio of the divergence in~\eqref{eqn:first_term} and Marton's exponent tends to one.
\end{proof}\vspace{-.05in}
\bibliographystyle{IEEEtran}
\bibliography{isitbib}
\end{document}

%% file: preamble.tex
\usepackage[mathscr]{eucal}
\usepackage[cmex10]{amsmath}
\usepackage{epsfig,epsf,psfrag}
\usepackage{amssymb,amsmath,amsthm,amsfonts,latexsym}
\usepackage{amsmath,graphicx,bm,xcolor,url}
\usepackage[caption=false]{subfig} 
\usepackage{fixltx2e}
\usepackage{array}
\usepackage{verbatim}
\usepackage{bm}
\usepackage{algorithmic, cite}
\usepackage{algorithm}
\usepackage{verbatim}
\usepackage{textcomp}
\usepackage{mathrsfs}
\usepackage{epstopdf}
\newcommand{\hatcalX}{\hat{\calX}}

\catcode`~=11 \def\UrlSpecials{\do\~{\kern -.15em\lower .7ex\hbox{~}\kern .04em}} \catcode`~=13 

\allowdisplaybreaks[4]
 
\newcommand{\nn}{\nonumber}

\newcommand{\calA}{\mathcal{A}}
\newcommand{\calB}{\mathcal{B}}
\newcommand{\calC}{\mathcal{C}}

\newcommand{\calE}{\mathcal{E}}

\newcommand{\calM}{\mathcal{M}}

\newcommand{\calP}{\mathcal{P}}

\newcommand{\calT}{\mathcal{T}}
\newcommand{\calU}{\mathcal{U}}

\newcommand{\calX}{\mathcal{X}}

\newcommand{\bbE}{\mathbb{E}}

\newcommand{\bbN}{\mathbb{N}}

\newcommand{\bbP}{\mathbb{P}}

\newcommand{\bbR}{\mathbb{R}}

\DeclareMathAlphabet{\mathbsf}{OT1}{cmss}{bx}{n}
\DeclareMathAlphabet{\mathssf}{OT1}{cmss}{m}{sl}

\DeclareSymbolFont{bsfletters}{OT1}{cmss}{bx}{n}  
\DeclareSymbolFont{ssfletters}{OT1}{cmss}{m}{n}
\DeclareMathSymbol{\bsfGamma}{0}{bsfletters}{'000}
\DeclareMathSymbol{\ssfGamma}{0}{ssfletters}{'000}
\DeclareMathSymbol{\bsfDelta}{0}{bsfletters}{'001}
\DeclareMathSymbol{\ssfDelta}{0}{ssfletters}{'001}
\DeclareMathSymbol{\bsfTheta}{0}{bsfletters}{'002}
\DeclareMathSymbol{\ssfTheta}{0}{ssfletters}{'002}
\DeclareMathSymbol{\bsfLambda}{0}{bsfletters}{'003}
\DeclareMathSymbol{\ssfLambda}{0}{ssfletters}{'003}
\DeclareMathSymbol{\bsfXi}{0}{bsfletters}{'004}
\DeclareMathSymbol{\ssfXi}{0}{ssfletters}{'004}
\DeclareMathSymbol{\bsfPi}{0}{bsfletters}{'005}
\DeclareMathSymbol{\ssfPi}{0}{ssfletters}{'005}
\DeclareMathSymbol{\bsfSigma}{0}{bsfletters}{'006}
\DeclareMathSymbol{\ssfSigma}{0}{ssfletters}{'006}
\DeclareMathSymbol{\bsfUpsilon}{0}{bsfletters}{'007}
\DeclareMathSymbol{\ssfUpsilon}{0}{ssfletters}{'007}
\DeclareMathSymbol{\bsfPhi}{0}{bsfletters}{'010}
\DeclareMathSymbol{\ssfPhi}{0}{ssfletters}{'010}
\DeclareMathSymbol{\bsfPsi}{0}{bsfletters}{'011}
\DeclareMathSymbol{\ssfPsi}{0}{ssfletters}{'011}
\DeclareMathSymbol{\bsfOmega}{0}{bsfletters}{'012}
\DeclareMathSymbol{\ssfOmega}{0}{ssfletters}{'012}

\newcommand{\hatx}{\hat{x}}
\newcommand{\hatX}{\hat{X}}

\newcommand{\tily}{\tilde{y}}
\newcommand{\tilY}{\tilde{Y}}

\newcommand{\hsigma}{\hat{\sigma}}

\def\fndot{\, \cdot \,}

\DeclareMathOperator*{\argmin}{arg\,min}

\DeclareMathOperator{\var}{\mathsf{Var}}
\DeclareMathOperator{\Vol}{Vol}

\newtheorem{theorem}{Theorem} 
\newtheorem{lemma}[theorem]{Lemma}

\newtheorem{definition}{Definition}

%% file: moderate_deviations.bbl
\begin{thebibliography}{10}
\providecommand{\url}[1]{#1}
\csname url@samestyle\endcsname
\providecommand{\newblock}{\relax}
\providecommand{\bibinfo}[2]{#2}
\providecommand{\BIBentrySTDinterwordspacing}{\spaceskip=0pt\relax}
\providecommand{\BIBentryALTinterwordstretchfactor}{4}
\providecommand{\BIBentryALTinterwordspacing}{\spaceskip=\fontdimen2\font plus
\BIBentryALTinterwordstretchfactor\fontdimen3\font minus
  \fontdimen4\font\relax}
\providecommand{\BIBforeignlanguage}[2]{{%
\expandafter\ifx\csname l@#1\endcsname\relax
\typeout{** WARNING: IEEEtran.bst: No hyphenation pattern has been}%
\typeout{** loaded for the language `#1'. Using the pattern for}%
\typeout{** the default language instead.}%
\else
\language=\csname l@#1\endcsname
\fi
#2}}
\providecommand{\BIBdecl}{\relax}
\BIBdecl

\bibitem{Sha59}
C.~E. Shannon, ``Coding theorems for a discrete source with a fidelity
  criterion,'' \emph{IRE Int. Conv. Rec.}, vol.~7, pp. 142–--163, 1959.

\bibitem{Marton74}
K.~Marton, ``Error exponent for source coding with a fidelity criterion,''
  \emph{IEEE Trans. on Inf. Th.}, vol.~20, no.~2, pp. 197–--199, Mar 1974.

\bibitem{Ihara00}
S.~Ihara and M.~Kubo, ``Error exponent of coding for memoryless gaussian
  sources with a fidelity criterion,'' \emph{IEICE Transactions}, vol. 83-A,
  no.~10, pp. 1891--–1897, 2000.

\bibitem{Iriyama05}
K.~Iriyama, ``Probability of error for the fixed-length lossy source coding of
  general sources,'' \emph{IEEE Trans. on Inf. Th.}, vol.~51, no.~4, pp.
  1498–--1507, Apr 2005.

\bibitem{PPV10}
Y.~Polyanskiy, H.~V. Poor, and S.~Verd\'{u}, ``Channel coding in the finite
  blocklength regime,'' \emph{IEEE Trans. on Inf. Th.}, pp. 2307 -- 59, May
  2010.

\bibitem{Hayashi09}
M.~Hayashi, ``Information spectrum approach to second-order coding rate in
  channel coding,'' \emph{IEEE Trans. on Inf. Th.}, pp. 4947 -- 66, Nov 2009.

\bibitem{TK12}
V.~Y.~F. Tan and O.~Kosut, ``On the dispersions of three network information
  theory problems,'' \emph{arXiv:1201.3901}, Feb 2012, [Online].

\bibitem{ingber11}
A.~Ingber and Y.~Kochman, ``The dispersion of lossy source coding,'' in
  \emph{Data Compression Conference (DCC)}, 2011.

\bibitem{Kos11}
V.~Kostina and S.~Verd\'{u}, ``Fixed-length lossy compression in the finite
  blocklength regime: Discrete memoryless sources,'' in \emph{Int. Symp. Inf.
  Th.}, 2011.

\bibitem{Kos11b}
------, ``Fixed-length lossy compression in the finite blocklength regime:
  {Gaussian} source,'' in \emph{Information Theory Workshop}, 2011.

\bibitem{Dembo}
A.~Dembo and O.~Zeitouni, \emph{Large Deviations Techniques and Applications},
  2nd~ed.\hskip 1em plus 0.5em minus 0.4em\relax Springer, 1998.

\bibitem{altug10}
Y.~Altug and A.~B. Wagner, ``Moderate deviation analysis of channel coding:
  Discrete memoryless case,'' in \emph{Int. Symp. Inf. Th.}, 2010.

\bibitem{PV10}
Y.~Polyanskiy and S.~Verd\'{u}, ``Channel dispersion and moderate deviations
  limits for memoryless channels,'' in \emph{Allerton Conference}, 2010.

\bibitem{Sason11}
I.~Sason, ``{On Refined Versions of the Azuma-Hoeffding Inequality with
  Applications in Information Theory},'' \emph{arXiv:1111.1977}, Nov 2011.

\bibitem{He09}
D.-K. He, L.~A. {Lastras-Monta\~{n}o}, E.-H. Yang, A.~Jagmohan, and J.~Chen,
  ``On the redundancy of {Slepian-Wolf} coding,'' \emph{IEEE Trans. on Inf.
  Th.}, vol.~55, no.~12, pp. 5607–--27, Dec 2009.

\bibitem{Csi97}
I.~Csisz\'ar and J.~Korner, \emph{Information Theory: Coding Theorems for
  Discrete Memoryless Systems}.\hskip 1em plus 0.5em minus 0.4em\relax
  Akademiai Kiado, 1997.

\bibitem{Tan11_IT}
V.~Y.~F. Tan, A.~Anandkumar, L.~Tong, and A.~S. Willsky, ``A large-deviation
  analysis for the maximum likelihood learning of {M}arkov tree structures,''
  \emph{IEEE Trans.\ on Inf.\ Th.}, vol.~57, no.~3, pp. 1714--35, Mar 2011.

\bibitem{Yu93}
B.~Yu and T.~P. Speed, ``A rate of convergence result for a universal
  d-semifaithful code,'' \emph{IEEE Trans. on Inf. Th.}, vol.~39, no.~3, pp.
  813–--820, Mar 1997.

\bibitem{Weiss03}
T.~Weissman, E.~Ordentlich, G.~Seroussi, S.~Verdu, and M.~L. Weinberger,
  ``{Inequalities for the $l_1$ deviation of the empirical distribution},''
  Hewlett-Packard Labs, Tech. Rep., 2003.

\bibitem{Zhang97}
Z.~Zhang, E.-H. Yang, and V.~K. Wei, ``The redundancy of source coding with a
  fidelity criterion: Known statistics,'' \emph{IEEE Trans. on Inf. Th.},
  vol.~43, no.~1, pp. 71–--91, Jan 1997.

\bibitem{VG05}
J.~L. Verger-Gaugry, ``{Covering a ball with smaller equal balls in
  $\bbR^n$},'' \emph{Disc.\ and Comp.\ Geom.}, vol.~33, no.~1, pp. 143–--155,
  2005.

\bibitem{Pal08}
H.~Palaiyanur and A.~Sahai, ``On the uniform continuity of the rate-distortion
  function,'' in \emph{Int. Symp. Inf. Th.}, 2008.

\end{thebibliography}
